\documentclass{llncs}
\usepackage{citesort}
\usepackage{graphicx}
\usepackage{theorem}
\usepackage{latexsym}
\usepackage{algorithm,algorithmic}
\usepackage{amssymb} 
\usepackage{multirow,eepic}

\makeatletter
\def\senbun#1(#2)#3({\@senbun(#2)(}
\def\@senbun(#1,#2)(#3,#4){%
   \@tempdima#1\p@ \advance\@tempdima#3\p@
   \divide\@tempdima\tw@
   \@tempdimb#2\p@ \advance\@tempdimb#4\p@
   \divide\@tempdimb\tw@
   \edef\@senbun@temp{\noexpand\qbezier(#1,#2)%
      (\strip@pt\@tempdima,\strip@pt\@tempdimb)(#3,#4)}%
   \@senbun@temp}
\makeatother

\newcommand{\N}{{\rm I\kern-.22em N}} 
\newcommand{\Z}{{\sf Z\kern-.42em Z}} 
\newcommand{\R}{{\rm I\kern-.22em R}} 
\newcommand{\BbbC}{{\rm\kern.22em\rule[.1ex]{.06em}{1.4ex}\kern-.28em C}} 
\newcommand{\BbbQ}{{\rm\kern.22em\rule[.1ex]{.06em}{1.4ex}\kern-.28em Q}}

\newcounter{Codeline}
\newcommand{\Newcodeline}{\setcounter{Codeline}{1}}
\newcommand{\Cl}{{\theCodeline}: \addtocounter{Codeline}{1}}
\newcommand{\crm}{\\}

\authorrunning{T. Okumura et.al.}
\titlerunning{Optimal Rendezvous for Asynchronous Mobile Robots with Lights}

\begin{document}

\title{Optimal Asynchronous Rendezvous \\for  Mobile Robots with Lights
}
\author{Takashi~OKUMURA\inst{1} \and Koichi~WADA\inst{2} \and Yoshiaki~KATAYAMA\inst{3}}
\institute{Graduate School of Science and Engineering, Hosei University, \\ 
Tokyo 184-8584 Japan.\\
\email{takashi.okumura.4e@stu.hosei.ac.jp}
\and
Faculty of Science and Engineering, Hosei University, \\Tokyo, 184-8485, Japan.\\ 
\email{wada@hosei.ac.jp}
\and
Graduate School of Engineering, Nagoya Institute of Technology, \\ Nagoya, 466-8555, Japan.\\
\email{katayama@nitech.ac.jp}
}
 
\maketitle

\begin{abstract}
We study a {\em Rendezvous} problem for $2$ autonomous mobile robots in asynchronous settings
with persistent memory called {\em  light}.
It is well known that Rendezvous is impossible when robots have no lights in basic common models,
even if the system is semi-synchronous. On the other hand,
Rendezvous is possible if robots have lights with a constant number of colors in several types lights\cite{FSVY,V}.
In asynchronous settings, Rendezvous can be solved by robots with $4$ colors of lights in non-rigid movement,
if robots can use not only own light but also other robot's light ({\em full-light}), where non-rigid movement means
robots may be stopped before reaching the computed destination but can move a minimum distance $\delta>0$ 
and rigid movement means robots can reach the computed destination.
In semi-synchronous settings, Rendezvous can be solved with $2$ colors of full-lights in non-rigid movement.
 
In this paper, we show that in asynchronous settings,
Rendezvous can be solved with $2$ colors of full-lights in rigid movement and
in non-rigid movement if robots know the value of the minimum distance $\delta$.
We also show that Rendezvous can be solved with $2$ colors of full-lights in general non-rigid movement
if we consider some reasonable restricted class of asynchronous settings. 


\end{abstract}

\section{Introduction}
\noindent{\bf Background and Motivation}\ \ 

The computational issues of autonomous mobile robots have been research object in distributed computing fields.
In particular, a large amount of work has been dedicated to the research of theoretical models of autonomous mobile robots
\cite{AP,BDT,CFPS,DKLMPW,IBTW,KLOT,SDY,SY}.
In the basic common setting, a robot is modeled as a point in a two dimensional plane and its capability is quite weak.
We usually assume that robots are {\em oblivious} (no memory to record past history), {\em anonymous}  and {\em uniform} (robots have no IDs and run identical algorithms)\cite{FPS}.
Robots operate in Look-Compute-Move (LCM) cycles in the model. In the Look operation, robots obtain a snapshot of the environment and
they execute the same algorithm with the snapshot as an input in the Compute operation, and move towards the computed destination in the Move operation.
Repeating these cycles, all robots perform a given task.
It is difficult for these too weak robot systems to accomplish the task to be completed. 
Revealing the weakest capability of robots to attain a given task is one of the most interesting challenges 
in the theoretical research of autonomous mobile robots.

The problem considered in this paper is {\em Gathering}, which is one of the most fundamental tasks of autonomous mobile robots. 
Gathering is the process of $n$ mobile robots, initially located on arbitrary positions, meeting within finite time at a location, not known a priori. 
When there are two robots in this setting, this task is called {\em Rendezvous}. 
In this paper, we focus on Rendezvous in asynchronous settings and we reveal the weakest additional assumptions for Rendezvous.

Since Gathering and Rendezvous are simple but essential problems, they have been intensively studied  
and a number of possibility and/or impossibility results have been shown under the different assumptions\cite{AP,BDT,CFPS,DGCMR,DKLMPW,DP,FPSW,IKIW,ISKIDWY,KLOT,LMA,P,SDY}.
The solvability of Gathering and Rendezvous depends on the activation schedule and the synchronization level. 
Usually three basic types of schedulers are identified, the fully synchronous (FSYNC), the semi-synchronous (SSYNC) and the asynchronous (ASYNC).
Gathering and Rendezvous are trivially solvable in FSYNC and the basic model.
However, these problems can not be solved in SSYNC without any additional assumptions \cite{FPS}. 



In \cite{DFPSY}, persistent memory called {\em light} has been introduced to reveal relationship between ASYNC and SSYNC and they show asynchronous robots with lights equipped with a constant number of colors, are strictly more powerful than semi-synchronous robots without lights.
In order to solve Rendezvous without any other additional assumptions, robots with lights have been introduced\cite{FSVY,DFPSY,V}.
Table~\ref{tab:Table-Rendezvous} shows results to solve Rendezvous by robots with lights in each scheduler and movement restriction.  
In the table,
{\em full-light} means that robots can see not only lights of other robots but also their own light, 
and {\em external-light} and {\em internal-light} 
mean that they can see only lights of other robots and only own light, respectively. 
In the movement restriction, 
Rigid means that robots can reach the computed destination.
In Non-Rigid,
robots may be stopped before reaching the computed destination but move a minimum distance $\delta>0$.
Non-Rigid(+$\delta$) means it is Non-Rigid and robots know the value $\delta$.
The Gathering of robots with lights is discussed in \cite{TWK}.
{\tiny
\begin{table}[h]
\centering
\caption{Rendezvous algorithms by robots with lights.}
\label{tab:Table-Rendezvous}
{\footnotesize
\begin{tabular}{|c|c|c|c|c|c|}
\hline
scheduler      & movement  & full-light & external-light & internal-light & no-light \\ \hline\hline
FSYNC & Non-Rigid
                            & \senbun(0,10)(38,-4) & \senbun(0,10)(56,-4) & \senbun(0,10)(56,-4) & $\bigcirc$  \\ \hline\hline
\multirow{3}{*}{SSYNC} & Non-Rigid     & 2 & 3        & ?        & \multirow{3}{*}{$\times$} \\ \cline{2-5} 
               & Rigid & \senbun(0,10)(38,-4)    & ?        & 6&     \\ \cline{2-5} 
               & Non-Rigid(+$\delta$) &\senbun(0,10)(38,-4)  & ?& 3 & \\ \hline\hline
\multirow{2}{*}{ASYNC}          & Non-Rigid     & 4    & ?       & ?        & \multirow{2}{*}{$\times$}  \\ \cline{2-5} 
               & Rigid & ?    & 12 & ?      &     \\ \cline{2-5}
                & Non-Rigid(+$\delta$) &?  & 3& ? & \\ \hline
\end{tabular}\\

Back slash indicates that  this part has been solved in a weaker condition.\\
$?$ menas this part is not solved.\\
}
\end{table}
}

\noindent{\bf Our Contribution}\ \ 

In this paper, we consider whether we can solve Rendezvous in ASYNC with the optimal number of colors of light.
In SSYNC,  Rendezvous cannot be solved  with one color but can be solved with $2$ colors in Non-Rigid and full-light.
On the other hand, Rendezvous in ASYNC can be solved with $4$ colors in Non-Rigid and full-light, 
with $3$ colors in Non-Rigid(+$\delta$) and external-light, or with $12$ colors in Rigid and internal-light, respectively.

In this paper we consider Rendezvous algorithms in ASYNC with the optimal number of colors of light and 
we show that Rendezvous in ASYNC can be solved with $2$ colors in Rigid and full-light, or
in Non-Rigid(+$\delta$) and full-light. We give a basic Rendezvous algorithm with $2$ colors of full-lights ($A$ and $B$) and 
it can solve Rendezvous in ASYNC and Rigid and its variant can also solve Rendezvous in ASYNC and Non-Rigid(+$\delta$).
These two algorithms can behave correctly if the initial color of each robot is $A$. However if the initial color of each robot is $B$,
the algorithm cannot solve Rendezvous in ASYNC and Rigid. 
It is still open whether Rendezvous can be solved with $2$ colors in ASYNC and Non-Rigid, however
we introduce some restricted class of ASYNC called  {\em LC-atomic} and 
we show that our basic algorithm can solve Rendezvous in this scheduler and Non-Rigid with arbitrary initial color, 
where LC-atomic ASYNC means we consider from the beginning of each Look operation to the end of the corresponding Compute operation as an atomic one, that is,  any robot cannot observe between the beginning of each Look operation and the end of each Compute one in every cycle. 
This is a reasonable sufficient condition Rendezvous is solved with the optimal number of colors of light in ASYNC and Non-Rigid.

The remainder of the paper is organized as follows. In Section
\ref{sec:model}, we define a robot model with lights, a Rendezvous problem, and terminologies. 
Section \ref{sec:PRR} shows 
the previous results for the Rendezvous problem,
and Section \ref{sec:RendezvousAlgorithms} shows Rendezvous algorithms of robots with lights 
on several situations of movement restriction. 
Section \ref{sec:conclusion}
concludes the paper.
 
\section{Model and Preliminaries}\label{sec:model}

We consider a set of $n$ anonymous mobile robots ${\cal R} = \{ r_1, \ldots, r_n \}$ located in $\R^2$.
Each robot $r_i$ has a persistent state $\ell(r_i)$ called {\em  light} which may be taken from a finite set of colors $L$. 

We denote by $\ell(r_i,t)$ the color of light the robot $r_i$ has at time $t$ and 
$p(r_i, t) \in \R^2$ the position occupied by $r_i$ at time $t$ represented in some global coordinate system. 
Given two points $p,q \in \R^2$, $dis(p,q)$ denotes the distance between $p$ and $q$.



Each robot $r_i$ has its own coordinate system where $r_i$ is located at its origin at any time. 
These coordinate systems do not necessarily agree with those of other robots. 
It means that there is no common unit of distance and no common knowledge of directions of its coordinates and clockwise orientation ({\em chirality}).


At any point of time, a robot can be active or inactive. When a robot $r_i$ is activated, it executes Look-Compute-Move operations:

\begin{itemize}
\item {\bf Look:} The robot $r_i$ activates its sensors to obtain a snapshot  which consists of pairs of a light and a position for every robot with respect to its own coordinate system. 
We assume robots can observe all other robots(unlimited visibility).
\item {\bf Compute:} The robot $r_i$ executes its algorithm using the snapshot and its own color of light (if it can be utilized) and returns a destination point $des_i$ by its coordinate system and a light $\ell_i \in L$ to which its own color is set.
\item {\bf Move:} The robot $r_i$ moves to the computed destination $des_i$.
The robot may be stopped by an adversary before reaching the computed destination. If stopped before reaching its destination, a robot moves at least a minimum distance $\delta >0$. If the distance to the destination is at most $\delta$, the robot can reach it. In this case, the movement is called {\em Non-Rigid}. Otherwise,
it is called {\em Rigid}. If the movement is Non-Rigid and robots know the value of $\delta$, it is called {\em Non-Rigid(+$\delta)$}. 
\end{itemize}




A scheduler decides which subset of  robots is activated for every configuration. 
The schedulers we consider are asynchronous and semi-synchronous and it is assumed that schedulers are {\em fair}, each robot is activated infinitely often.
\begin{itemize}
\item {\bf ASYNC:} 
The asynchronous (ASYNC) scheduler, activates the robots independently, and the duration of each Compute, Move and between successive activities is finite and unpredictable. As a result, robots can be seen while moving and the snapshot and its actual configuration are not the same and so its computation may be done with the old configuration.
\item {\bf SSYNC:} 
The  semi-synchronous(SSYNC) scheduler activates a subset of all robots synchronously  and their Look-Compute-Move cycles are performed at the same time. We can assume that activated robots at the same time obtain the same snapshot and their Compute and Move are executed instantaneously.
In SSYNC, we can assume that each activation defines discrete time called {\em round} and Look-Compute-Move is performed instantaneously  in one round. 
\end{itemize}

As a special case of SSYNC, if all robots are activated in each round,
the scheduler is called full-synchronous (FSYNC).

In this paper, we consider ASYNC and we assume the followings;

In a Look operation, a snapshot of a time $t_L$ is taken and we say that {\em Look operation is performed at time $t_L$.}
Each Compute operation of $r_i$ is assumed to be done at an instant time $t_C$ and its color of light $\ell_i(t)$ and its destination $des_i$ are assigned to the computed values at the time $t_C$. 
In a Move operation, when its movement begins at $t_B$ and ends at $t_E$, we say that its movement is performed during $[t_B.t_E]$,
its beginning and ending of the movement are denoted by $Move_{BEGIN}$ and $Move_{END}$,
and its $Move_{BEGIN}$ and $Move_{END}$ occur at $t_B$ and $t_E$, respectively.  
In the following, $Compute$, $Move_{BEGIN}$ and $Move_{END}$ are abbreviated as $Comp$, $Move_{B}$ and $Move_{E}$, respectively.
When some cycle has no movement (robots change only colors of lights, or their destinations are the current positions),
we can assume the Move operation in this cycle is omitted, since we can consider 
the Move operation can be performed just before the next Look operation. 

Also we consider the following restricted classes of ASYNC;

Let a robot execute a cycle.
If any other robot cannot execute any Look operation between the Look operation and the following Compute one in the cycle,
its ASYNC is said to be {\em $LC$-atomic}. Thus we can assume that in LC-atomic ASYNC,
Look and Compute operations in every cycle  are performed at the same time. 
If any other robot cannot execute any Look operation between the $Move_B$ and the following $Move_E$,
its ASYNC is said to be {\em $Move$-atomic}.
In this case Move operations in all cycles can be considered to be performed instantaneously and at time $t_M$.
In Move-atomic ASYNC, when a robot $r$ observes another robot  $r'$ performing a Move operation at time $t_M$,
$r$ observes the snapshot after the moving of $r'$.



In our settings, robots have  persistent lights and can change their colors at an instant time in each Compute operation. 
We consider the following robot models according to visibility of lights.
\begin{itemize}
\item {\em full-light}, the robot can recognize not only colors of lights of other robots but also its own color of light.
\item {\em external-light}, the robot can recognize only colors of lights of other robots but cannot see its own color of light. 
Note robot can change its own color.
\item {\em internal-light}, the robot can recognize only its own color of light but cannot see colors of lights of other robots.
\end{itemize}

An $n$-{\em Gathering} problem is defined that given $n (\geq 2)$ robots initially  placed in arbitrary positions  in $\R^2$,
they congregate at a single location which is not predefined in finite time.
In the following, we consider the case that $n=2$ and
the $2$-Gathering problem is called {\em Rendezvous}.

\section{Previous Results for Rendezvous}
\label{sec:PRR}
Rendezvous is trivially solvable in FSYNC but is not in SSYNC.

\begin{theorem} \label{theorem:Rand_Impo}
{\em \cite{FPS}} Rendezvous is deterministically unsolvable in SSYNC even if chirality is assumed.
\end{theorem}

If robots have a constant number of colors in their lights, Rendezvous can be solved shown in the following theorem (or Table~\ref{tab:Table-Rendezvous}).

\begin{theorem}
{\em \cite{FSVY,DFPSY,V}}
\begin{enumerate}
\item[(1)] Rendezvous is solved in full-light, Non-Rigid and SSYNC with $2$ colors.
\item[(2)] Rendezvous is solved in external-light, Non-Rigid and SSYNC with $3$ colors.
\item[(3)] Rendezvous is solved in internal-light, Rigid and SSYNC with $6$ colors.
\item[(4)]  Rendezvous is solved in internal-light, Non-Rigid(+$\delta$) and SSYNC with $3$ colors.
\item[(5)] Rendezvous is solved in full-light, Non-Rigid and ASYNC with $4$ colors.
\item[(6)] Rendezvous is solved in external-light, Rigid and ASYNC with $12$ colors.
\item[(7)] Rendezvous is solved in external-light, non-rigid(+$\delta$) and ASYNC with $3$ colors.
\end{enumerate} 
\end{theorem}
It is still an open problem that Rendezvous is solved in ASYNC with $2$ colors.
In the following, we will show that Rendezvous is solved  in ASYNC and full-light with $2$ colors,
if we assume (1) Rigid movement, (2) Non-Rigid movement and knowledge of the minimum distance $\delta$ robots move, (3) LC-atomic.  
In these cases, we can construct optimal Rendezvous algorithms with respect to the number of colors in ASYNC.

\section{Asynchronous Rendezvous Algorithms for Robots with Lights}
\label{sec:RendezvousAlgorithms}

\subsection{Basic Rendezvous Algorithm}

In this section, two robots are denoted as $r$ and $s$.
Let $t_0$ be the starting time of the algorithm. 

Given a robot $robot$, an operation $op$($\in \{ Look, Comp, Move_{B}, Move_{E} \}$), and a time $t$,
$t^+(robot, op)$ denotes the time $robot$ performs $op$ immediately after $t$ if there exists such operation, and
$t^-(robot, op)$ denotes the time $robot$ performs $op$ immediately before $t$  if there exists such operation.
If $t$ is the time the algorithm terminates, $t^+(robot, op)$ is not defined for any $op$.
When $robot$ does not perform $op$ before $t$ and $t^-(robot, op)$ does not exist, $t^-(robot, op)$ is defined to be $t_0$.

A time $t_c$ is called a {\em cycle start time}, if the next performed operations of both $r$ and $s$ after $t$ are Look ones, or
otherwise, the robots performing the operations neither change their colors of lights nor move.
In the latter case, we can consider that these operations can be performed before $t_c$ and the subsequent $Look$ operation can be performed as the first operation after $t_c$.

\Newcodeline
\begin{algorithm}[h]
\caption{Rendezvous (scheduler, movement, initial-light)}
\label{algo:Ren}
{\footnotesize
\begin{tabbing}
111 \= 11 \= 11 \= 11 \= 11 \= 11 \= 11 \= \kill
{\em Parameters}: scheduler, movement-restriction, Initial-light \crm
{\em Assumptions}: full-light, two colors ($A$ and $B$) \crm

\Cl \> {\bf case} me.light   {\bf of } \crm

\Cl \> $A$: \crm
\Cl \> \> {\bf if} other.light =$A$ {\bf then} \crm
\Cl \> \> \>$me.light \leftarrow B$ \crm
\Cl \> \> \>$me.des \leftarrow$ the midpoint of $me.position$ and $other.position$\crm
\Cl \> \>{\bf else} $me.des \leftarrow other.position$ \crm 
\Cl \> $B$: \crm
\Cl \>\> {\bf if} $other.light = A$ {\bf then} \crm
\Cl \>\>\>$me.des \leftarrow me.position$ // stay\crm
\Cl \> \> {\bf else} $me.light \leftarrow A$ \crm 

\Cl \> {\bf endcase} 
\end{tabbing}
}
\end{algorithm}

Algorithm~\ref{algo:Ren} is used as a basic Rendezvous algorithm 
which has three parameters, schedulers, movement restriction and an initial color of light and
assumes full-light and uses two colors $A$ and $B$. 

We will show that Rendezvous(ASYNC, Rigid, A) and Rendezvous(LC-atomic ASYNC, Non-Rigid, any) solve Rendezvous 
and some variant of Rendezvous(ASYNC, Non-Rigid(+$\delta$), A) also solves Rendezvous.

Algorithm~\ref{algo:Ren} behaves as follows.

When both colors of $r$ and $s$ are $A$, they change their colors into $A$ and they move to the midpoint of the two current positions,
when one's (say $r$) color is $A$ and the other's ($s$) color is $B$, $s$ stays at the current position and $r$ moves to the position, and
when both colors of the two robots are $B$, they change their colors into $A$.

It is easily verified that Rendezvous(SSYNC, Non-Rigid ,any) solves Rendezvous.
However, it is not trivial to prove whether this algorithm works well in ASYNC or not. 
In fact,
Rendezvous(ASYNC, Rigid, B) can not work correctly, which we will show later.

The following two lemmas are useful for proving the correctness.
The first one is easily verified. and
note they hold for Non-Rigid movement.

\begin{lemma}\label{BBdontmove}
Assume that time $t_c$ is a cycle start time and $\ell(r,t_c)=\ell(s,t_c)=B$ in Rendezvous(ASYNC, Non-Rigid, any). 
If $dis(p(r,t_c)),p(s,t_c))=0$,
then two robots $r$ and $s$ do not move after $t_c$.
\end{lemma}

\begin{lemma}\label{ABRendezvous}
Let robot $r$ perform Look operation at time $t$ in Rendezvous(ASYNC, Non-Rigid, any).
If $t^-(s, Comp) \leq t$ and $\ell(r,t) \neq \ell(s,t)$,
then there exists a time $t^* (> t)$ such that $r$ and $s$ succeed  in rendezvous at time $t^*$ by Rendezvous(ASYNC, Non-Rigid, any).
\end{lemma}
\begin{proof}

If $\ell(r,t)=B$ and $\ell(s,t)=A$, then $r$ does not change the color and stays at the current position.
If $s$ performs a Look operation at $t^+(s, Look)$, $s$ does not change the color and the $s$' destination is $p(r,t)$.
Since both $r$ and $s$ do not change the colors after the time $t^+(s, Look)$,
$r$ stays at $p(r,t)$ and the destination of $s$ is $p(r,t)$.
Thus $r$ and $s$ succeed in rendezvous at some time $t^* \geq t^+(s, Move_{E})$ even in Non-Rigid.

If $\ell(r,t)=A$ and $\ell(s,t)=B$, then $r$ does not change the color and computes the destination as $p(s,t)$.
When $s$ finishes the Move operation at $t'=t^+(s,Move_{E})$, $s$ is located  at $p(s,t')$. 
If $t' \leq t$, since $r$'s destination is $p(s,t)$ and $p(s,t)$ is not changed (even if $s$ performs $Look$ operation after $t'$ and before $t$), 
$r$ and $s$ succeed in rendezvous at some time $t^* \geq t^+(r, Move_{E})$. 

Otherwise ($t <t'$), 
if $r$ performs Look operations before $t'$, these destinations are different because $s$ is moving, but the color is not changed and $\ell(r,t')=A$.
Since  $s$ stays at $p(s,t')$ after $t'$, $r$ and $s$ succeed in rendezvous at some time $t^* \geq t^+(r, Move_{E})$. 

In both cases $r$ and $s$ do not move after $t^*$ by the algorithm.
\end{proof}

\subsubsection{ASYNC and Rigid movement \\}

If Rigid movement is assumed, asynchronous Rendezvous can be done with $2$ colors.
\begin{figure*}
 \begin{center}
    \includegraphics[scale=0.8]{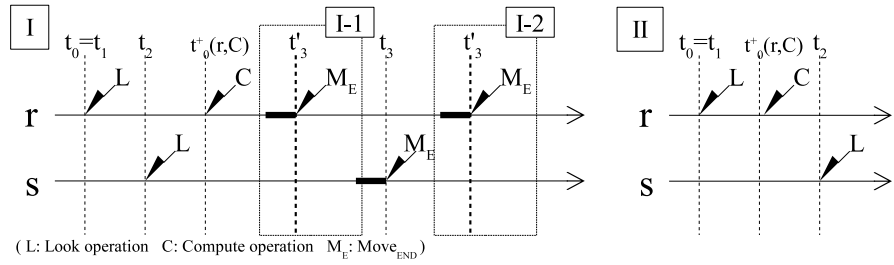}
    \caption{Several cases in the proof of Theorem~\ref{AsyncRigidA}}
    \label{figCaseTheorem3}
  \end{center}
\end{figure*}

\begin{theorem}
\label{AsyncRigidA}
Rendezvous(ASYNC, Rigid, A) solves Rendezvous.
\end{theorem}
\begin{proof} 
Let $t_0$ be the starting time of the algorithm and let $r$ and $s$ be two robots whose colors of lights are $A$.
Without loss of generality, $r$ is assumed to perform $Look$ operation first at time $t_1$, that is, $t_1=t_0^+(r, Look)$.
Let $t_2=t_0^+(s,Look)$ and there are two cases;
(I) $t_1 \leq t_2 < t_0^+(r, Comp)$, (II) $t_0^+(r, Comp) \leq t_2$ (Fig.~\ref{figCaseTheorem3}).
\vspace{3mm}

(I) Since $\ell(s,t_2)=A$ and $\ell(r,t_2)=A$, $s$ move to the midpoint of $p(r,t_0)$ and $p(s,t_0)$ at time $t_3=t_0^+(s,Move_{E})$.
Robot $s$ changes its color of light from $A$ to $B$ at time $t_0^+(s, Comp)$ and $\ell(s,t_3)=B$.
There are two cases (I-1) $t_3'=t_0^+(r, Move_{E}) < t_3$ and (I-2) $t_3 \leq t_3'=t_0^+(r, Move_{E})$.

(I-1) Robot $r$ reaches to the destination at time $t_3'$ but $s$ does not reach the destination ($t_3'=t_0^+(r, Move_{E}) < t_3$).
If $r$ does not perform any operations during $[t_3', t_3]$, $t_3$ becomes a cycle start time and 
then $r$ and $s$ rendezvous  at time $t_3$ and the two robots do not move after $t_3$ by Lemma~\ref{BBdontmove}.

Otherwise, $r$ performs several operations during $[t_3', t_3]$. 
If $r$ performs at least one $Look$ operation and one $Comp$ operation during $[t_3', t_3]$ and let $t_C$ be the time $r$ performs $Comp$ operation ($t_C=t_3'^+(r,Comp)$).
Note that $r$ only changes its color of light and does not move in this cycle. 
Then its color of light is changed to A at $t_C$ and $\ell(r,t_C)=A$ and $\ell(s,t_C)=B$. 
Thus, the next Look operation of $r$ or $s$ after $t_C$ satisfies the conditions of  Lemma~\ref{ABRendezvous}, $r$ and $s$ succeed in rendezvous.
The remaining case is that $r$ performs only $Look$ operation during $[t_3', t_3]$.
Let $t_{L}$ be the time $r$ performs the $Look$ operation.
Since $r$ observes $\ell(s,t_{L})=B$ and
$r$ and $s$ are located at the same point at $t_3$,
this case is the same as the first case.

(I-2) Interchanging the roles of $r$ and $s$, this case can be reduced to (I-1).

(II) Since $(t_0)^+(r, Comp) \leq t_2$ and $\ell(r,t_2) \neq \ell(s,t_2)$, $r$ and $s$ succeed in rendezvous by Lemma~\ref{ABRendezvous}.
\end{proof}

Note that this algorithm does not terminate and we cannot change the algorithm so that the fixed one can terminate.
It is an open problem whether there exists an algorithm which solves Rendezvous and terminates with two colors in ASYNC.  
Also there exists an execution that Rendezvous(Async, Rigid, any) does not work in general. In fact,
if initial colors of lights for both robots are B, this algorithm cannot solve Rendezvous.
Fig.~\ref{figExecutionAsyncRigidB} shows a counterexample Rendezvous(Async, Rigid, B) does not work.
Since the colors of lights at $t_5$ are B, this execution repeats forever and achieves only convergence, that is, the robots move arbitrarily close to each other, but might not rendezvous within finite time.
This counterexample also shows Rendezvous(Move-atomic ASYNC, Rigid, B) does not work.
However, if we assume LC-atomic ASYNC, we can show that Rendezvous(LC-atomic ASYNC, Rigid, B) solves Rendezvous.

\begin{figure*}
 \begin{center}
    \includegraphics{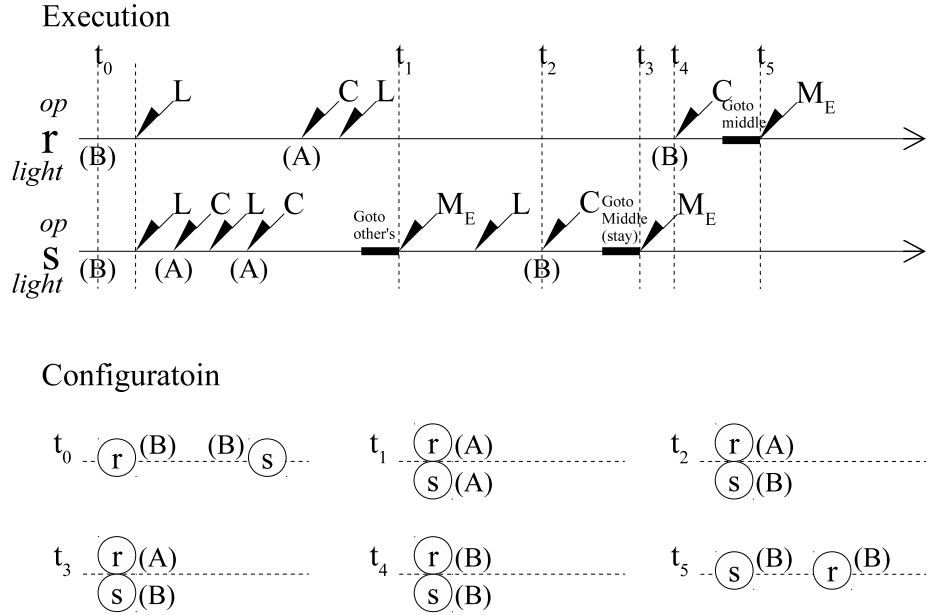}
    \caption{Rendezvous(ASYNC, Rigid, B) cannot solve Rendezvous in general.}
    \label{figExecutionAsyncRigidB}
  \end{center}
\end{figure*}

\begin{lemma} \label{RendInitB}
Rendezvous(LC-atomic ASYNC, Rigid, B) solves Rendezvous.
\end{lemma}
\begin{proof}
In LC-atomic ASYNC, any Look operation and the following Comp operation are performed at the same time and this operation is denoted as LC.
Let $t_c$ be a cycle start time and let $r$ perform an LC operation first and let $t_1=t_c^+(r,LC)$.
There are two cases;
(I) $t_1=t_c^+(s,LC)$, and (II) $t_1< t_c^+(s,LC)=t_2$.

(I) In this case, since $t_1$ becomes a cycle start time and $\ell(r,t_1)=\ell(s,t_1)=A$,  this lemma holds by Theorem~\ref{AsyncRigidA}.

(II) In this case, since $\ell(r,t_1)=A$ and $\ell(s,t_2)=B$,  this lemma holds by Lemma~\ref{ABRendezvous}.
\end{proof}
Since it is easily verified that there is a cycle start time $t_c (\geq t_0)$ such that $\ell(r,t_c)=\ell(s,t_c)=B$ 
in an execution of Rendezvous(Async, Non-Rigid, A), it cannot solve Rendezvous even if both initial colors of lights are $A$. 
In the next subsection, we will show if ASYNC is restricted to LC-atomic one,
Rendezvous can be solved in Non-Rigid with two colors from any initial colors of lights.

\subsubsection{LC-atomic ASYNC and Non-Rigid movement\\}

Let $t_c$ be a cycle start time of the algorithm.
There are three cases  according to the colors of lights of two robots $r$ and $s$,
(I) $\ell(r,t_c) \neq \ell(s,t_c)$, (II) $\ell(r,t_c)=\ell(s,t_c)=A$, and (III) $\ell(r,t_c) = \ell(s,t_c)=B$

\begin{lemma}
\label{LCatomicAB}
If $\ell(r,t_c) \neq \ell(s,t_c)$ and the algorithm starts at $t_c$, then  there exists a time $t^* (\geq t)$ such that $r$ and $s$ succeed  in rendezvous 
at time $t^*$ by Rendezvous(LC-atomic ASYNC, Non-Rigid, any).
\end{lemma}
\begin{proof}
It is obvious from Lemma~\ref{ABRendezvous}.
\end{proof}

\begin{lemma}
\label{LCatomicAA}
If $\ell(r,t_c)=\ell(s,t_c)=A$ and the algorithm starts at $t_c$, then there exists a time $t^* (\geq t_c)$ such that $r$ and $s$ succeed  in rendezvous 
at time $t^*$ by Rendezvous(LC-atomic ASYNC, Non-Rigid, any) or $t^*$ is a cycle start time, $\ell(r,t^*)=\ell(s,t^*)=A$ and $dis(p(r,t^*),p(s,t^*)) \leq dis(p(r,t_c),p(s,t_c))-2\delta$.
\end{lemma}
\begin{proof}
Let $r$ perform the $LC$ operation first and let $t_1=t_c^+(r,LC)$.
There are two cases;
(I) $t_1=t_c^+(s,LC)$, and (II) $t_1< t_c^+(s,LC)=t_2$.

(I)  In this case, $\ell(r,t_1)=\ell(s,t_1)=B$ and there are two cases;
(I-1) $t_1^+(r,LC) \neq t_1^+(s,LC)$, and
(I-2) $t_1^+(r,LC) =t_1^+(s,LC)$.

(I-1) This case can be proved by Lemma~\ref{ABRendezvous}.

(I-2) Letting $t^*=t_1^+(r,LC) =t_1^+(s,LC)$, $\ell(r,t^*)=\ell(s,t^*)=A$ and  $t_*$ becomes a cycle start time.
Also at the time $t_*$ rendezvous is succeeded or $dis(p(r,t^*),p(s,t^*)) \leq dis(p(r,t_c),p(s,t_c))-2\delta$.

(II) Since $\ell(r,t_1)=A$ and $\ell(s,t_2)=B$, this case is proved by Lemma~\ref{ABRendezvous}.
\end{proof}

\begin{lemma}
\label{LCatomicBB}
If $\ell(r,t_c)=\ell(s,t_c)=B$ and the algorithm starts at $t_c$, then  there exists a time $t^* (\geq t)$ such that $r$ and $s$ succeed  in rendezvous 
at time $t^*$ by Rendezvous(LC-atomic ASYNC, Non-Rigid, any) or $t^*$ is a cycle start time and $\ell(r,t^*)=\ell(s,t^*)=A$.
\end{lemma}
\begin{proof}
Let $t_1=t_c^+(r,LC)$ and $t_c^+(s,LC)$.
If $t_1 =t_2$, then letting $t^*=t_1$, $t^*$ is a cycle start time and $\ell(r,t^*)=\ell(s,t^*)=A$.
Otherwise, Lemma~\ref{ABRendezvous} proves this case.
\end{proof}

Lemmas~\ref{RendInitB}-\ref{LCatomicBB} is followed by the next theorem.

\begin{theorem}
\label{AsyncLCatomicNonRigidAny}
Rendezvous(LC-atomic ASYNC, Non-Rigid, any) solves Rendezvous.
\end{theorem}

\subsubsection{ASYNC and Non-Rigid movement(+$\delta$)\\}

Although it is still open whether asynchronous Rendezvous can not be solved in Non-rigid with two colors of lights,
if we assume Non-Rigid($+\delta$), we can solve Rendezvous  modifying  Rendezvous(ASYNC, Non-Rigid($+\delta$), A) and
using the minimum moving value $\delta$ in it.

\Newcodeline
\begin{algorithm}[h]
\caption{RendezvousWithDelta (ASYNC, Non-Rigid($+\delta$), $A$)}
\label{algo:RenWDelta}
{\footnotesize
\begin{tabbing}
111 \= 11 \= 11 \= 11 \= 11 \= 11 \= 11 \= \kill
{\em Assumptions}: full-light, two colors ($A$ and $B$) \crm\crm

\Cl \> {\bf case} $dis(me.position, other.position)(=DIST)$ {\bf of } \crm

\Cl \> $DIST>2\delta$: \crm
\Cl \> \> {\bf if} me.light =other.light =$B$ {\bf then} \crm
\Cl \> \> \>$me.des \leftarrow$ the point moving  by $\delta/2$ from $me.position$ to $other.position$\crm
\Cl \> \>{\bf else} $me.light \leftarrow B$ \crm 
\Cl \> $2\delta \geq DIST \geq \delta$: \crm
\Cl \>\> {\bf if} $me.light=other.light = A$ {\bf then} \crm
\Cl \>\>\>$me.light \leftarrow B$ \crm
\Cl \>\>\>$me.des \leftarrow$ the midpoint of $me.position$ and $other.position$\crm
 \Cl \> \> {\bf else} $me.light \leftarrow A$ \crm 
\Cl \> $\delta>DIST$:  //Rendezvous(ASYNC, Rigid, A) \crm

\Cl \>\> {\bf case} me.light   {\bf of } \crm

\Cl \> \>$A$: \crm
\Cl \> \>\> {\bf if} other.light =$A$ {\bf then} \crm
\Cl \> \> \>\>$me.light \leftarrow B$ \crm
\Cl \> \> \>\>$me.des \leftarrow$ the midpoint of $me.position$ and $other.position$\crm
\Cl \> \>\>{\bf else} $me.des \leftarrow other.position$ \crm 
\Cl \>\> $B$: \crm
\Cl \>\>\> {\bf if} $other.light = A$ {\bf then} $me.des \leftarrow me.position$ // stay\crm
\Cl \> \> \>{\bf else} $me.light \leftarrow A$ \crm 

\Cl \> \>{\bf endcase} \crm
\Cl \> {\bf endcase} 
\end{tabbing}
}
\end{algorithm}

Let $dist_0 = dis(p(r,t_0),p(s,t_0))$ and let
RendezvousWithDelta (Algorithm~\ref{algo:RenWDelta})  begin
with $\ell(r,t_0)=\ell(s,t_0)=A$. If $dist_0 >2\delta$, both robots do not move until both colors of lights become $B$({\bf lines} 3-5)
and there exists a cycle start time $t_1(> t_0)$ such that $\ell(r,t_1)=\ell(s,t_1)=B$.
After $\ell(r,t_1)=\ell(s,t_1)=B$, the distance between $r$ and $s$ is reduced by $\delta/2$ without changing the colors of lights({\bf line} 4) and
the distance falls in $[2\delta, \delta]$ and both colors of lights become $A$ at a cycle starting time $t_2$.
After $\ell(r,t_2)=\ell(s,t_2)=A$, we can use Rendezvous(ASYNC, Rigid, A) since $2\delta \geq dis(p(r,t_2),p(s,t_2)) \geq \delta$. Therefore,
rendezvous is succeeded.
Note that in Algorithm~\ref{algo:RenWDelta}, the initial pair of colors of $r$ and $s$ is $(\ell(r,t_0),\ell(s,t_0))=(A,A)$ and
it is changed into $(\ell(r,t_1),\ell(s,t_1))=(B,B)$ without changing the distance of $r$ and $s$. And it is changed into $(\ell(r,t_2),\ell(s,t_2))=(A,A)$
when the distance becomes between $\delta$ and $2\delta$. These {\em mode} changes are necessary and our algorithm does not work correctly, if these mode changes are not incorporated in the algorithm.
 
 \begin{figure*}
 \begin{center}
    \includegraphics[scale=0.8]{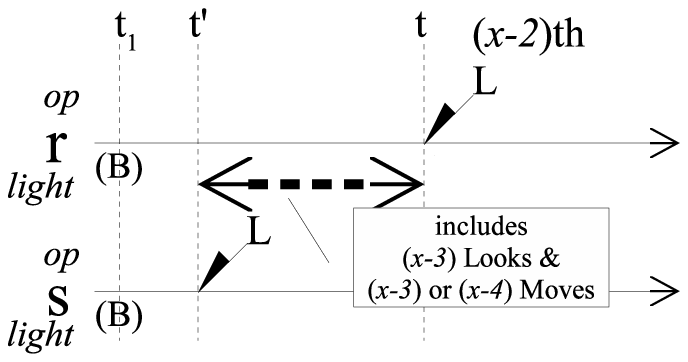}
    \includegraphics[scale=0.8]{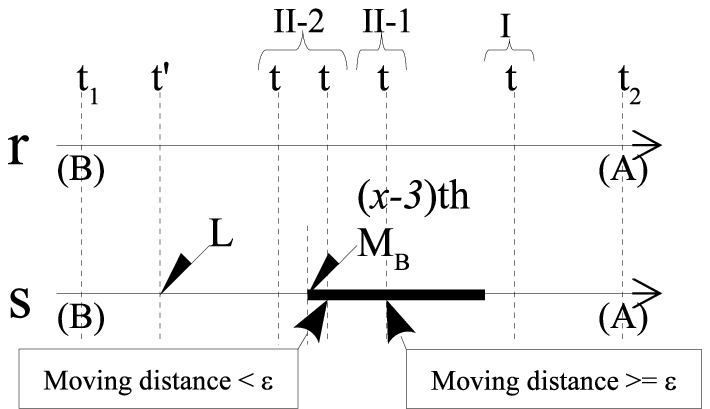}
    \caption{Situations in the proof of Lemma~\ref{distgeq2delta}}
    \label{figCaseLemma7}
  \end{center}
\end{figure*} 

\begin{lemma}\label{distgeq2delta}
If $dist_0 > 2\delta$,  in any execution of RedezvousWithDelta(ASYNC, Non-Rigid($+\delta$), $A$), 
\begin{enumerate}
\item[(1)] there exists a cycle start time $t_1(> t_0)$ such that $\ell(r,t_1)=\ell(s,t_1)=B$ and $dis(p(r,t_1),p(s,t_1))=dist_0$, and
\item[(2)] there exists a cycle start time $t_2 (> t_1)$ such that $\ell(r,t_2)=\ell(s,t_2)=A$ and $2\delta \geq dis(p(r,t_2),p(s,t_2)) \geq \delta$.
\end{enumerate}
\end{lemma}
\begin{proof}
{\em (1)} Without loss of generality, $r$ performs the $Look$ operation first and let $t_{rL}$ be such time. 
The color of $r$ is changed from $A$ to $B$ at a time $t_{rC}$.
Since $\ell(s, t_0)=A$, $s$ performs a $Look$ operation at a time $t_{sL} (\geq t_{rL})$ and changes its color from $A$ to $B$ at a time $t_{sC}$.  
Then, let $t_1=max(t_{rC},t_{sC})$. If a $Comp$ operation is performed immediately after $t_1$, the robot does not change its color of light,
since the robot performs the preceding $Look$ operation before $t_1$.
Thus, $t_1$ becomes a cycle start time.

{\em (2)} Since $t_1$ is a cycle start time, we can consider that
the algorithm starts at $t_1$ with $\ell(r,t_1)=\ell(s,t_1)=B$.
The distance $dist_0$ is reduced by $\delta/2$ every one cycle of each robot after $t_1$.
Since $dist_0 > 2\delta$, $dist_0$ can be denoted as $x(\delta/2)+\epsilon$, where $x \geq 4$ and $0 \leq \epsilon < \delta/2$.

Let $t$ be a time of the $(x-2)$-th $Look$ operation among $Look$ operations $r$ and $s$ performed  after $t_1$ and without loss of generality,
let $r$ be the robot performing the $(x-2)$-th $Look$ operation. 
Note that among $(x-3)$ $Move$ operations between $t_1$  and $t$ at least $max(0,x-4)$ $Move$ operations have been completed and at most one $Move$ operation has not completed yet.

Let $t'=t^-(s,Look)$\footnote{If $s$ performed no $Look$ operations after $t_1$, $t'=t_1$.}.
We have
two situations(Fig.~\ref{figCaseLemma7}). One is that (I) $(x-3)$ $Move$ operations are completed until $t$. This case satisfies $2\delta \geq dis(p(r,t),p(s,t)) \geq \delta$.
The other is that (II) the $(x-3)$-th $Move$ operation $s$ performs has not been completed at $t$
\footnote{This case includes $t < t'^+(s, Move_{B})$}. The latter case is divided into two cases 
 (II-1) $2\delta  \geq dis(p(r,t),p(s,t)) \geq \delta$ and
(II-2) $dis(p(r,t),p(s)) > 2\delta$
according to the time $r$ performs the $Look$ operation.

{\bf Case (I) and (II-1):} 
Since $2\delta \geq dis(p(r,t),p(s,t)) \geq \delta$, $r$ changes its color of light to $A$ at $t^+(r,Comp)$.
When $s$ performs a $Look$ operation at $t_{sL}=t^+(s,Look)$, $s$ observes  $2\delta \geq dis(p(r,t_{sL}),p(s,t_{sL})) \geq \delta$ and $\ell(s,t_{sL})=B$ and
changes its color of light to $A$ at $t_{sC}=t^+(s,Comp)$.  
Letting $t_2=max(t_{rC},t_{sC})$, $t_2$ becomes a cycle start time as follows.

When $t_2=T_{rC}$, $s$ changes its color of light to $A$ at $t_{sC} ( \leq t_{rC})$. Even if 
$s$ performs a $Look$ operation at $t_L$ after $t_{sC}$ before $t_2=t_{rC}$, $s$ does not change its color at $t_L^+(s,Comp)$ since $\ell(r,t_L)=B$.
The case that $_2=T_{sC}$ can be proved similarly.

{\bf Cases (II-2):} 
Since $dis(p(r,t),p(s,t)) > 2\delta$, $r$ reduces the distance by $\delta/2$.
Then, $r$ performs the $Move$ operation and 
subsequently performs the next $Look$ operation at $t_{rL}=t^+(r, Look)$ then changes its color of light to $A$ at $t_{rC}=t^+(r,Comp)$,
since $\delta \leq dis(p(r,t_{rL}),p(s,t_{rL})) \leq 2\delta$. 
The next $Look$ operation of $s$ is performed after $t'^+(s,Move_E)$ and $s$
changes its color of light to $A$ at $t_{sC}=t^+(s,Comp)$. 
Robot $s$ changes its color of light to $A$ at $t_{sC}$.
Letting $t_2=max(t_{rC},t_{sC})$, we can prove that $t_2$ becomes a cycle start time similar to the former case. 
%

%
\end{proof}

The followintg two lemmas can be proved similar to the proof of Theorem~\ref{AsyncRigidA}.

\begin{lemma}\label{distgeqdelta}
If $2\delta \geq dist_0 \geq \delta$, then RedezvousWithDelta(ASYNC, Non-Rigid($+\delta$), $A$) solves Rendezvous.
\end{lemma}

\begin{lemma}\label{distleqdelta}
If $dist_0 >\delta$, then RedezvousWithDelta(ASYNC, Non-Rigid($+\delta$), $A$) solves Rendezvous.
\end{lemma}

Lemmas~\ref{distgeq2delta}-\ref{distleqdelta} imply the following theorem.

\begin{theorem}\label{ASYNCNonRigidDeltaA}
RedezvousWithDelta(ASYNC, Non-Rigid($+\delta$), $A$) solves Rendezvous.
\end{theorem}

\section{Concluding Remarks}
\label{sec:conclusion}

We have shown that Rendezvous can be solved in ASYNC with the optimal number of colors of lights if Rigid  or Non-Rigid($+\delta$) movement is assumed.
We have also shown that Rendezvous can be solved in ASYNC and Non-Rigid with  the optimal number of colors of lights if ASYNC is LC-atomic.
Although we conjecture that Rendezvous cannot be solved in ASYNC and Non-Rigid with $2$ colors,
it is still open whether it can be solved or not.

\paragraph*{Acknowledgment}
This work is supported in part by KAKENHI no. 17K00019 and 15K00011.


\newpage
\appendix

\end{document}